\documentclass[12pt,a4paper]{amsart}
\usepackage{amsmath}
\usepackage{amsfonts}
\usepackage{natbib}
\usepackage{graphicx}
\usepackage{amssymb}
\usepackage{bm}
\usepackage{amsthm}
\usepackage{tikz}

\newtheorem{thm}{Theorem}
\newtheorem{lem} [thm]{Lemma}
\newtheorem{cor}[thm]{Corollary}

\newcommand{\cT}{{\mathcal T}}

\newcommand{\R}{\mathbb R}
\newcommand{\T}{\mathcal T}

\newtheorem{prop}[thm]{Proposition} 
\newtheorem*{prop1}{Proposition~\ref{prophelps}}
\newtheorem*{thm2}{Theorem~\ref{thm}}

\title[Tree-like Reticulation Networks]{Tree-like Reticulation Networks -\\ When Do Tree-like Distances Also Support Reticulate Evolution?}

\author{Andrew R. Francis}
\address{Centre for Research in Mathematics,\\ 
University of Western Sydney, Australia. }
\email{a.francis@uws.edu.au}

\author{Mike Steel}
\address{Biomathematics Research Centre,\\ 
University of Canterbury,  New Zealand.}
\email{mike.steel@canterbury.ac.nz}

\begin{document}

\begin{abstract}
Hybrid evolution and horizontal gene transfer (HGT) are processes where evolutionary relationships may more accurately be described by a reticulated network than by a tree. In such a network, there will often be several paths between any two extant species, reflecting the possible pathways that genetic material may have been passed down from a common ancestor to these species. These paths will typically have different lengths but an `average distance' can still be  calculated between any two taxa. In this  article, we ask whether this average distance is able to distinguish reticulate evolution from pure tree-like evolution. We consider two types of reticulation networks:  hybridization networks and HGT networks. For the former, we establish a general result which shows that average distances  between extant taxa can appear tree-like, but only under a single hybridization event near the root; in all other cases, the two forms of evolution can be distinguished by average distances.  For HGT networks, we demonstrate some analogous but more intricate results.  

\end{abstract}

\maketitle

\section{Introduction}
Evolutionary relationships between present-day taxa (species, genera etc) are usually represented by a phylogenetic tree, which shows a branching pattern of speciation from some ancestral taxon to the taxa we observe today \citep{fel}. However, reticulate evolution is known to complicate this simple `tree model'  due to processes such as the formation of hybrid species \citep{mcb}, and other mechanisms where genetic material is exchanged between species (such as horizontal gene transfer (HGT)) or within a species (recombination, a process we do not consider further in this paper).  Consequently, phylogenetic networks  that allow `vertical' branching through time as well as `horizontal' reticulation events have increasingly been recognised as providing a more complete picture of much of the evolutionary 
history of life \citep{hus, hus2, nak}. 

This transition has brought with it a number of mathematical and computational problems -- in particular, how to reconstruct and analyse such networks, and how to distinguish different types of reticulation from tree-like evolution \citep{hold, holl}. In this note we consider one aspect of the latter topic, namely the question of whether or not, if we knew the average evolutionary distance between each pair of species, we could determine whether the species network could have been a tree, or whether some  more complicated reticulate history is required.  

In a phylogenetic tree, the evolutionary distance between two present-day species is simply the path length from each species to the other via its most recent common ancestor (here, `evolutionary distance' typically refers to the actual or expected amount of genetic change).  However, for networks, there may be many paths linking two present-day species, and the evolutionary distance will be some average of these path lengths. Nevertheless, it is conceivable that in some cases, these distances might still appear to fit a tree exactly.  We explore this question for two classes of networks: those relevant to 
hybrid evolution; and those relevant to HGT. Both are special cases of a more general description of (binary) `reticulation' networks, which we now define.

\subsection{Definitions: Reticulation Networks}
Following~\citet{linz}, a {\em reticulation network} $N$ on a finite set $X$  is a rooted acyclic digraph $(V,A)$ with the following properties:
\begin{itemize} 
\item[(i)]  the {\em root} vertex has in-degree 0 and out-degree 2;
\item[(ii)]  $X$ is the set of vertices with out-degree 0 and in-degree 1 (`leaves');
\item[(iii)] all remaining vertices are {\em interior vertices}, and each such vertex either has in-degree 1 and out-degree 2  (a {\em tree vertex})
 or in-degree 2 and out-degree 1 (a {\em reticulation vertex});
\item[(iv)]   the arc set $A$ of $N$ is the disjoint union of two subsets, the set of `reticulation arcs' $A_R$ and the set of `tree arcs' $A_T$; moreover each reticulation arc ends at a reticulation vertex, and each reticulation vertex has at least one incoming reticulation arc;
\item[(v)]  every interior vertex has at least one outgoing tree arc; and
\item[(vi)] there is a function $t: V \rightarrow \R$ so that (a) if $(u,v)$ is a tree arc then $t(u)<t(v)$, and (b) if $(u,v)$ is a reticulation arc, then $t(u)=t(v)$. 
\end{itemize}

Condition (vi) embodies the biological requirement that the network has a temporal representation that reflects the order of speciation events, and for which reticulation events involve  two species that co-exist at some point in time. 

In applications, $X$ typically denotes a set of extant (present day) species.   Two types of reticulation networks are particularly relevant in evolutionary biology (for different reasons, as we explain shortly) and these will be the main classes we will consider in this paper.  The distinction is in the pair of  arcs ending at a reticulation vertex in property (iii).  Namely,
\begin{itemize}
\item
in a  {\em hybridization network}, \underline{both} arcs ending in a reticulation vertex are reticulation arcs, and
\item in a {\em horizontal gene transfer (HGT) network}, exactly \underline{one} of the arcs ending in a reticulation vertex is a reticulation arc.
\end{itemize}
A simple example of each type is shown in Figure~\ref{fig1}.

\begin{figure}[htb]
\centering
\includegraphics[scale=1.0]{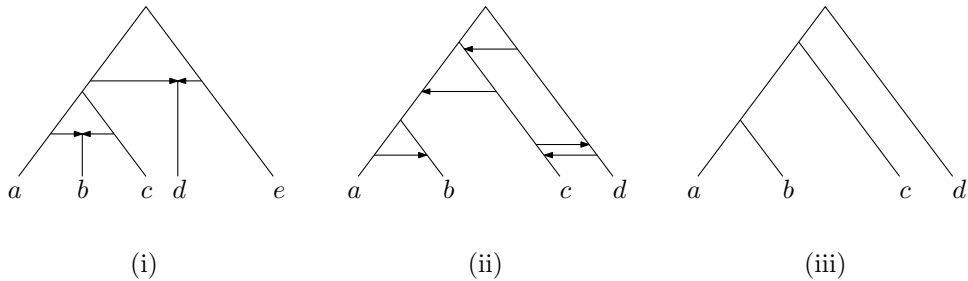}
\caption{(i) A hybridization network on $\{a,b,c,d,e\}$ (usually extant species); (ii) an HGT network on $\{a,b,c,d\}$; and (iii)  the tree $T_N$ obtained from the HGT network $N$ in (ii) by deleting all reticulation arcs.  Reticulate arcs in (i) and (ii) are drawn as arrows; in each case the reticulate vertices are at the endpoints of the reticulate arcs. Note that (i) has four reticulation arcs and two reticulation vertices, while (ii) has five reticulation arcs and five reticulation vertices.  }
\label{fig1}
\end{figure}

Hybridization networks model settings where a new species arises from members of two lineages, a process that occurs in plants, fish, and some animals~\citep{bul,mcb},  while HGT models the situation where a gene (or genes) are transferred from one species to another (a process that is common in bacteria) \citep{dag}.

\section{Reticulation Networks and Average Distances}

\subsection{Basic Properties of Reticulation Networks}
Firstly, observe that a reticulation network $N$ on $X$ has no reticulation vertices if and only if $N$ is a rooted binary phylogenetic $X$-tree (as defined, for example, in \citet{sem}).

Moreover, any  hybridization network  is necessarily a {\em tree-child network};  that is, from any interior vertex in $N$, there is a path to a leaf that avoids any reticulation vertex. Tree-child networks have a  number of desirable combinatorial and computational properties (see e.g. \citet{car, van}).

Hybridization networks have bounded size once $n=|X|$ is specified, since such a network can have at most $n-2$ reticulation vertices \citep{mcd}.   To see this, note that in any digraph, the sum of the out-degrees equals the sum of the in-degrees so we obtain:
\begin{equation}
\label{counter}
2+2t+r =\sum_{v \in V} {\rm deg_{out}}(v) =  \sum_{v \in V} {\rm deg_{in}}(v) = n+t+2r,
\end{equation}
 where  $t$ and $r$ refer to the number of tree vertices and hybridization vertices, respectively. 
Note that each hybridization vertex corresponds to two parent tree vertices, and hence $t\ge 2r$ in a hybridization network.   
Eqn. (\ref{counter}) gives $n = t+2-r$, and using $t \geq 2r$ we obtain:
\begin{equation}
\label{counter2}
r \leq n-2.
\end{equation}
A consequence of this bound is that, up to isomorphism, there are only finitely many hybridization networks for any given $n$ (the enumeration of hybridization networks has recently been investigated by \citet{mcd}).

By contrast, an HGT network with a given leaf set $X$ can have arbitrarily many reticulation vertices, and so there are infinitely many  HGT networks for a given $X$.
However, an HGT network $N$ has a useful property that is absent in a hybridization network:  an HGT network always has an
 associated canonical rooted binary phylogenetic $X$-tree $T$ that is obtained from $N$ by deleting all the reticulation arcs (and suppressing any resulting vertices that have both 
 in-degree 1 and out-degree 1). 
 We denote this tree with the notation $T_N$ (an example is shown in Figure~\ref{fig1}).

Given any reticulation network $N$ on $X$, suppose that for each reticulation vertex, we delete exactly one of the in-coming arcs. The resulting graph is a rooted tree with leaf set $X$ and a root that coincides with the root of $N$. Moreover, if we suppress any resulting vertices that have both in-degree 1 and out-degree 1 we obtain a rooted binary phylogenetic $X$-tree, $T$. We say that $T$ is {\em displayed} by $N$ and we let $\T(N)$ denote the set of all the (at most) $2^r$ such trees that are displayed by $N$.

\subsection{Tree Metrics}\label{sec:tree.metrics}
Consider any unrooted phylogenetic $X$-tree $T=(V,E)$ together with a weight function $w:E\to\R^{>0}$ that assigns strictly positive weights to each edge of the tree.
Then $(T, w)$ induces a distance function on $X$
as follows:   For each pair of leaves $x,y$ on a tree $T$, the {\em tree distance} between them is defined as the sum of the weights of the edges that lie on the (unique) path in $T$ connecting $x$ and $y$. That is: 
\[d_{(T,w)}(x,y):=\sum_{e\in P(T;x,y)} w(e),\]  
where if $x=y$ we set $d_{(T,w)}(x,y) =0$ (the empty path has length zero). The resulting function $d_{(T,w)}:X\times X\to \R^{\ge 0}$ is a metric on $X$.

A metric on $X$ that can be represented in this way on some phylogenetic $X$-tree is said to be a {\em tree metric}. This holds  if and only if the metric satisfies the  `four-point condition'. This states that for any four (not necessarily distinct) points $u,v,w,y$ from $X$, two of the three
sums $d(u,v)+d(w,y); d(u,w)+d(v,y); d(u,y) + d(v,w)$ are equal, and are greater than or equal to the other one. 
This classic characterisation of tree metrics dates back to the 1960s (for more recent treatments, see  \citet{dre, sem}).
Moreover, if $d$ is a tree metric on $X$,  then $d$ can be written $d=d_{(T, w)}$ for precisely one choice of the pair $(T, w)$, where $T$ is a phylogenetic $X$-tree, and $w$ a strictly positive edge weight function. In the  case where $T$ is binary,  we will say that $d$ is a {\em binary tree metric}.

\subsection{Average Distances on Networks}\label{avsec}
A reticulation network can be thought of as a `weighted union' of the trees displayed by $N$.  We formalise this idea, and extend it to bring in distances,  as follows: 

For each vertex $v$ in the set $V_R$ of reticulation vertices of $N$, let $R(v)$ denote the two arcs that end at $v$.
Suppose we  are given a reticulation network $N=(V,A)$ on $X$ along with:
\begin{itemize}
\item a weight function $w: A_T \rightarrow \R^{>0}$ that assigns weights to each tree arc; 
and
\item a strictly positive probability distribution $\beta$ on the set $F_N$ of functions $f: V_R \rightarrow A$ for which $f(v) \in R(v)$.  
\end{itemize}

In evolutionary biology, the weighting $w$ typically describes some measure of genetic change along each tree arc,  and each function $f \in F_N$ indicates the line of descent of a particular gene, and so describes a tree $T_f$ in $\T(N)$. Notice that $|F_N|=2^r$ though it may be possible for different functions $f$ to lead to the same rooted phylogenetic tree (possibly with different tree metrics).  For a given $f \in F_N$, its $\beta$-value, denoted $\beta_f$,  can be thought of as the expected proportion of genes that follow the tree $T_f$.
Since $\sum_f \beta_f = 1$, we call $\beta$ the `mixing distribution' of the network.

For example,  suppose we have two arcs $a$ and $a'$ that end at the reticulate vertex $v$, and a function 
 $\alpha: \{a, a'\} \rightarrow \R^{>0}$ satisfying $\alpha(a)+ \alpha(a')=1$  (such a function $\alpha$ 
could indicate the proportion of genetic material 
that is contributed from each of the two parent lineages when a reticulation occurs).  
When there is just a single reticulation, the mixing distribution $\beta$ can be identified with the $\alpha$ function for the single reticulation vertex.   However, when more than one reticulation  vertex is present, one needs to consider how the different reticulation events might interact.  
In the simplest case, we might treat  the reticulations as (stochastically) independent events  (with $\alpha$ now being regarded as assigning probabilities rather than proportions)
so that the resulting randomly generated function $f$ would have probability $\beta_f = \prod_{v \in V_R} \alpha(f(v)).$   This assumption of independence is very strong and is more than we require here. Indeed, all that we require is that the mixing distribution $\beta$  of the network satisfies $\beta_f>0$ for all $f \in F_N$.

We now define the distance induced by a reticulation network with weighted tree arcs and a mixing distribution.  
For such a triple $(N, w, \beta)$ we define
$$d= d_{(N, w, \beta)}: X \times X \rightarrow \R^{\geq 0}$$
by 
$$d(x,y) = \sum_{f \in F_N} \beta_f d_{(T_f, w_f)} (x,y),$$
where $w_f$ is the edge weight induced by $N$ on $T_f$.
If  there are no reticulation vertices in $N$, so that $N$ is a rooted phylogenetic $X$-tree $T$, we take $d$ to be the tree metric $d_{(T,w)}$. 
We illustrate these ideas for a hybridization network and an HGT network in Figures \ref{figure2} and~\ref{figure3}, respectively.

Notice that $d_{(N, w, \beta)}$ is always a metric on $X$, since it is a convex combination of tree metrics on $X$.

\begin{figure}[ht]
\centering
\includegraphics[scale=1.0]{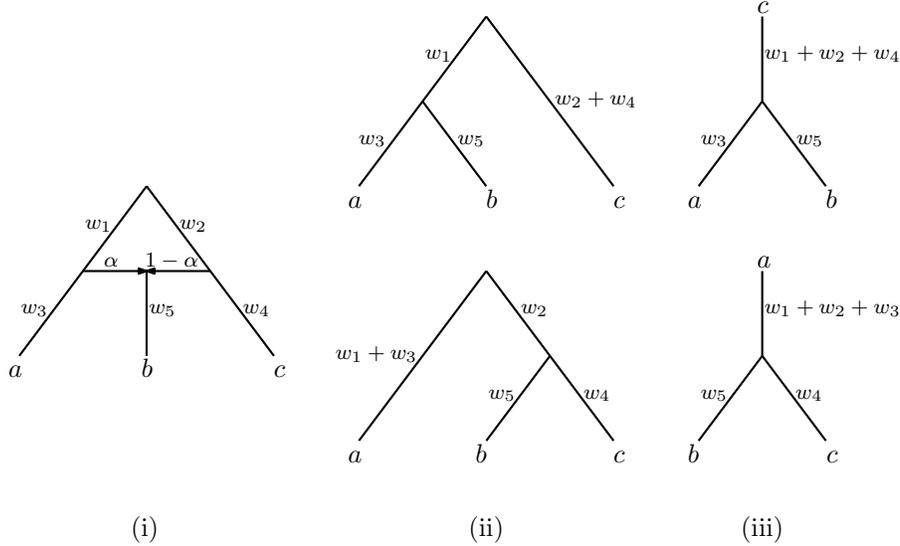}
\caption{(i) A simple hybridization network $N$ with edge weights $w$, and reticulation values $\alpha,1-\alpha$. (ii)  The two rooted trees in $\T(N)$ that are displayed by $N$, together with their associated
edge weights. (iii) The unrooted trees from (ii), which have the same topology, even though the trees in (ii) do not.   For this example, the network distance between $a$ and $b$ is given by $d_N(a,b) = \alpha(w_3+w_5) + (1-\alpha)(w_1+w_2+w_3+w_5).$}
\label{figure2}
\end{figure}

\begin{figure}[ht]
\centering
\includegraphics[scale=1.0]{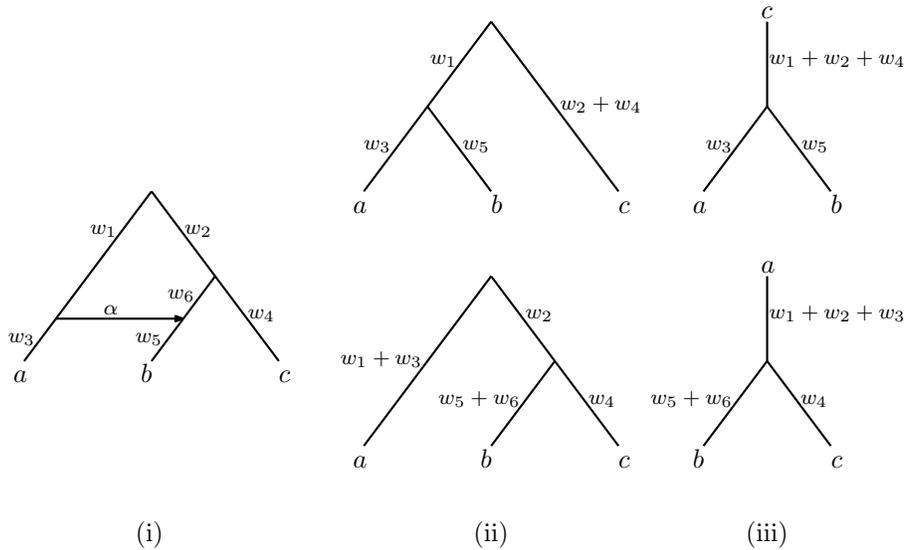}
\caption{(i) A simple HGT network with edge weights $w$, and reticulation value $\alpha$ on the unique reticulation arc, and $1-\alpha$ on the incident tree arc that has weight $w_6$.
(ii) The two rooted trees in $\T(N)$ that are displayed by $N$, together with their associated
edge weights. (iii) The unrooted trees from (ii).  For this example, the network distance between $a$ and $b$ is given by $d_N(a,b) = \alpha(w_3+w_5) + (1-\alpha)(w_1+w_2+w_3+w_5+w_6).$}
\label{figure3}
\end{figure}

Some detail is known about the conditions that govern when a metric can be represented on a reticulation network.    In \cite{willson2012}, Willson shows how a network can be reconstructed from average distances, given that one knows the underlying network graph already.  There are some uniqueness properties, including for the reticulation probabilities ($\alpha$).  

In~\cite{willson2013}, Willson shows how one can generate the underlying network graph from the average distances, under some hypotheses.  He shows that if there is only one reticulation, it can be done, and he provides an algorithm for this.  There are some necessary conditions on the average distance function, and which are sufficient if there is only one reticulation.

In this paper, we study conditions under which a tree metric can be represented on a network.  We impose no conditions on the distance function.

Network distances arise in a range of models in molecular genetics, for example those in which  DNA sequences  undergo site mutations along the tree arcs, and for which
(i) at speciation events (tree vertices) the sequences on the two outgoing arcs are identical to the sequence at the end of the incoming arc,
and (ii) at reticulation vertices, the state at each site is selected from the state at the same position at the end of either one of the two incoming arcs.
In such a hybridisation network $N$ on $X$, the history of the $i^{\rm th}$ position for each species in $X$ traces back according to one of the trees in $\cT(N)$.  Now, consider the Hamming distance $d_H$ between pairs of species from the set $X$ (so $d_H(x,y)$ is the proportion of sites
 where taxon $x$ and taxon $y$ differ). For binary sequences, suppose each site mutates at most once in the network (the so-called `infinite sites model'  \cite{dur}) and at each reticulation vertex the state at each site is (independently) selected to match the state at the site from one of its two incoming arcs with the prescribed $\alpha$ probability values. Then  the expected Hamming distance $\overline{d_H}$ on $X$ satisfies $\overline{d_H}= d_{(N, w, \beta)}$, where the weight $w$ of a tree arc $(u,v)$ is the proportion of sites for which a mutation occurs along $(u,v)$.

\subsection{Tree Metrics from a Network}
We will show (Proposition~\ref{prophelps})  that if each tree in $\T(N)$ is isomorphic to the same (unrooted) phylogenetic $X$-tree, then  the network induces a distance that is tree-like and behaves nicely with respect to the weights.  On the other hand, if exactly two different unrooted trees are  present in $\T(N)$,  then the distance function induced by the network is never tree-like. 
The proof of this result can be found in the Materials and Methods section. 

\begin{prop}
\label{prophelps}
\begin{itemize}
\item[(a)]
Suppose that  all the trees in $\T(N)$ are isomorphic as unrooted phylogenetic $X$-trees to some tree $T$. Then $d_N$ is a tree metric that is represented by $T$.
\item[(b)]If the trees in $\T(N)$ can be partitioned into two non-empty isomorphism classes of unrooted trees, then $d_N$ is not a tree metric. 
\end{itemize}
\end{prop}

\section{Hybridization Networks}
In this section, we are  interested in whether or not a tree metric can be realised on a hybridization network and, conversely, whether a hybridization network might induce a distance that fits perfectly on some tree.  In order to state our main result, Theorem~\ref{thm}, we introduce a further definition:
we call a hybridization network with $k$ hybridizations a \emph{$k$-hybridization network}, and we call a 1-hybridization network for which the two reticulation arcs have 
their source vertices adjacent to the root a \emph{primitive} 1-hybridization network.
The proof of this Theorem is given in the Materials and Methods section.

\begin{thm}\label{thm}
Let $X$ be a finite set of taxa, and suppose $d$ is a metric on $X$, $d:X\times X\to\R^{\ge 0}$.  
\begin{enumerate}
\item[(a)] If $d$ is a binary tree metric, then there exists a primitive 1-hybridization network $N$ and weights $w,\beta$ such that $d=d_{(N, w, \beta)}$.
\item[(b)] If $N$ is a hybridization network, and  $d=d_{(N, w, \beta)}$ is a tree metric for some $w,\beta>0$, then $N$ is either a tree, or $N$ is a primitive 1-hybridization network. 
\end{enumerate}
\end{thm}

One way to rephrase the key point of this theorem is that a tree metric can be represented on a $k$-hybridization network if and only if either $k=0$ or $k=1$ and the hybridization is placed near the root.  But there is slightly more here, in that any network admitting a tree metric must be primitive 1-hybridization network.  We are also able to count such networks, showing  that there are $4(n-3)$ such primitive 1-hybridization networks on $n$ taxa for each tree metric (Proposition~\ref{prop:counting}).

\subsection{The Number of Hybridization Networks for each Tree Metric}
There is a unique unrooted tree for each tree metric.  Theorem~\ref{thm} means that for each rooted tree we have at least one primitive 1-hybridization network.  But how many do we in fact have?

\begin{prop}\label{prop:counting}
For each tree metric on $n$ leaves, there are $4(n-3)$ 1-hybridization networks that realise the metric. 
\end{prop}

\begin{proof}
Each of the $n-3$ internal edges on the unrooted tree correspond to a choice for the root, and each one gives four distinct 1-hybridization networks, since each of the four subtrees nearest the root could be the one descending from the hybridization (see Figure~\ref{fig:four.retics}).
Each hybridization network from a root placed on an external edge can  also be obtained from a root placed on an internal edge, and so this adds no new hybridization networks.  
\end{proof}

\begin{figure}[ht]
\centering
      \newdimen\nodeSizeb
      \nodeSizeb=6mm
\resizebox{.9\textwidth}{2cm}{
\includegraphics{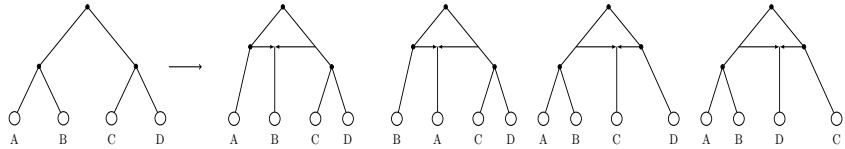}
}
\caption{The four distinct reticulated trees arising from choices of reticulation on a single internal edge of $T$, shown with the root placed on it on the left hand figure.  A, B, C and  D represent  subtrees of $T$ with the root placed at an internal node. }\label{fig:four.retics}
\end{figure}

\section{HGT Networks}
Our main result for hybridization networks (in the previous section) applies only in one direction for HGT networks.  If  a single reticulation occurs between the arcs of $N$ that are incident with the root 
then we obtain a tree metric. However, for HGT networks,  it is possible for tree metrics to arise under other scenarios, both for a single reticulation event, and for multiple ones.  We now describe two results that demonstrate how this can occur. Recall that $T_N$ is the rooted phylogenetic $X$-tree associated with an HGT network $N$, obtained by deleting all the reticulation arcs.  

\begin{lem}
\label{lemsister}
For any HGT network $N$, if each reticulation arc is  between adjacent tree arcs of $T_N$, then $d_N$ is tree-like on $T_N$.
\end{lem} 
\begin{proof}
If $N$ has the property described,  then every tree in  $\T(N)$ is isomorphic as a rooted phylogenetic $X$-tree to $T_N$, and so these two trees are also isomorphic as unrooted trees.  The lemma now follows from Proposition \ref{prophelps}(a).
\end{proof}

Our main result for this section is the following.

\begin{thm}
\label{main_HGT}
\mbox{ } 
\begin{itemize}
\item[(a)] If an HGT network  $N$ has a single reticulation arc, then $d_N$ is tree-like if and only if that arc is either (i) from one arc to an adjacent arc or (ii) between a root arc and one of the two children of the other root arc.  Moreover, this holds for any (positive) parameters on $N$, and the only tree that harbours a representation for $d_N$ is $T_N$. 
\item[(b)] There exist 2-reticulated HGT networks $N$ that can be represented on $T_N$ and (for other parameter settings) on a tree that is different from $T_N$, even when the mixing distribution treats the two reticulations independently.
\end{itemize} 
\end{thm}
 
\begin{proof}
{\em Part (a):}  For the `if' part, condition (i) suffices by Lemma~\ref{lemsister}.  For case (ii), we note that although the two trees in $\T(N)$ are no longer isomorphic to the same rooted phylogenetic $X$-tree, they are isomorphic to the same unrooted phylogenetic $X$-tree, so Proposition~\ref{prophelps} applies.    
 
For the `only if' direction, suppose that neither condition (i) nor (ii) is satisfied.  That is, the reticulation arc is not between adjacent arcs and not from a root arc to one of the two children of the other arc.  There is a quartet then in which the reticulation is between non-adjacent and non-root arcs, in which case, if we suppress the location of the root, it corresponds to the scenario shown in  Figure~\ref{fig:hgt.quartets}, up to permutation of the leaves.

Let us abbreviate the sums of distances arising in the four-point condition as $S_1=d(1,2)+d(3,4)$, $S_2=d(1,3)+d(2,4)$ and $S_3=d(1,4)+d(2,3)$.
Ignoring the terms that appear in every sum (shown as $\ast$ in Figure~\ref{fig:hgt.quartets}),  the quartet distance sums in the case shown in Figure~\ref{fig:hgt.quartets} are:
\begin{align*}
S_1&=a+[\alpha(a+b)+(1-\alpha)c],\\
S_2&=[(1-\alpha)(b+c)+\alpha a]+(a+b),\\
S_3&=b+[(1-\alpha)(a+b+c)].
\end{align*}
Noting that $S_1< S_2$ since $\alpha<1$, for these quartets to satisfy the four-point condition we must have $S_2=S_3$.  However, this implies that either $a=0$ or $\alpha=0$, which is a contradiction. 

\begin{figure}[ht]
\centering
\footnotesize
\includegraphics{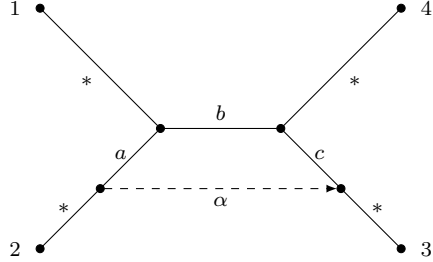}
\caption{The generic case of a single HGT from a tree arc to a non-adjacent tree arc.  The rooting of the tree has been suppressed to simplify the analysis; however, there are six locations where the root can be placed to subdivide the tree arcs shown (the arc labelled $*$ leading to leaf $3$ cannot contain the root, as this would create a directed cycle in the network,  but any other tree arc can).    Here $\ast$ denotes weights that occur in each quartet sum in the four-point condition and that hence can be ignored. }
\label{fig:hgt.quartets}
\end{figure}

{\em Part (b):}
It suffices to provide an example. Consider the 2-reticulated network shown in Figure~\ref{fig:2-retic}, where HGT events represented by the branches labelled $\alpha$ and $\alpha'$ occur independently (the \emph{independence model}, as described in the section Average Distances on Networks).

\begin{figure}[ht]
\centering
\footnotesize
\includegraphics{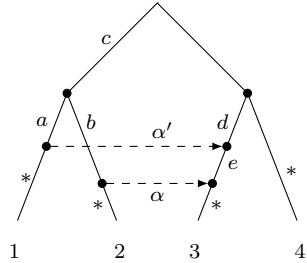}
\caption{A 2-reticulated HGT network $N$ that can be represented on a tree that is different from $T_N$.  Here, we assume that $a,b,c,d,e>0$ and $0<\alpha,\alpha'<1$.}\label{fig:2-retic}
\end{figure}

We have the following quartet distances involved in the four-point condition, with $S_i$ being as defined in (a):
\begin{align*}
S_1	&=[a+b]+[(1-\alpha)(1-\alpha')(d+e)+\alpha(b+c)+(1-\alpha)\alpha'(a+c+e)]\\
	&=(1+(1-\alpha)\alpha')a+(1+\alpha)b+(\alpha+(1-\alpha)\alpha')c+(1-\alpha)(1-\alpha')d+(1-\alpha)e,\\
S_2	&=[(1-\alpha)(1-\alpha')(a+c+d+e)+(1-\alpha)\alpha' e+\alpha(a+b)]+[b+c]\\
	&=(1-(1-\alpha)\alpha')a+(1+\alpha)b+(1+(1-\alpha)(1-\alpha'))c+(1-\alpha)(1-\alpha')d+(1-\alpha)e,\\
S_3	&=[a+c]+[(1-\alpha)\alpha'(a+b+e)+(1-\alpha)(1-\alpha')(b+c+d+e)]\\
	&=(1+(1-\alpha)\alpha')a+(1-\alpha)b+(1+(1-\alpha)(1-\alpha'))c+(1-\alpha)(1-\alpha')d+(1-\alpha)e.
\end{align*}

In the underlying tree ($T_N$) of the network we have $S_1$ as the smaller of these, so that $S_1\le S_2=S_3$.  The equality of $S_2$ and $S_3$ requires $\alpha b=(1-\alpha)\alpha' a$, and $S_1\le S_2$ implies $\alpha' a\le (1-\alpha')c$.  Together we require 
\[\alpha b=(1-\alpha)\alpha' a\le (1-\alpha)(1-\alpha')c,\]
which is certainly possible for some regions of the parameter space.

However, there are alternative solutions, as required by the theorem.  For instance, it is possible to have $S_3$ as the shortest of the three quartet distances, so that $S_3\le S_2=S_1$.  This is possible so long as
\[(1-\alpha)(1-\alpha')c=(1-\alpha)\alpha' a\le \alpha b.\]
Just to be explicit, this is possible whenever, for example, $ \alpha\ge \frac{1}{2}$, $ b\ge a$ and $\alpha' a = (1-\alpha') c$.
The unrooted tree that realizes this metric has taxa 1 and 4 together and  taxa 2 and 3 together ($14|23$), and is \emph{not} $T_N$.
This completes the proof of Theorem~\ref{main_HGT}.
\end{proof}

\section{Discussion and Further Questions}
The four point condition provides a very precise characterization, in terms of pairwise distances between taxa, of the circumstances under which a metric is able to be displayed on a tree (see the section on Tree Metrics).  It is so successful that it is tempting to assume that once a metric satisfies this condition then we have a tree, and that that is the end of the story. However the results in this paper show that such ``tree metrics" can \emph{also} be realised as hybridization and HGT networks.  Any surprise at this conclusion may be partly due to the biconditional statement of the four point condition; namely that a metric is a tree metric \emph{if and only if} it satisfies the condition.  Superficially this appears to leave little room to maneouvre.  However, as we show, being realised on a tree does not preclude the possibility that the metric can also be realised on a reticulation network.

The practical implication of this wriggle-room is that phylogenies displaying tree metrics may in fact involve hybridization or horizontal gene transfer in their histories.  However, the results in this paper also show that `all hell is not about to break loose': for the network to be a hybridization network, strict restrictions apply (Theorem~\ref{thm}).  In particular, there can be at most one hybridization event and it must be adjacent to the root.  However, such restrictions do not hold for HGT networks (Theorem~\ref{main_HGT}).  In this case there is some control when the network contains a single reticulation, but surprisingly, it is also possible to have a tree metric displayed on an HGT network with more than one reticulation.  

While it is biologically unlikely for a single network to contain both hybridization and HGT events, these results leave open several intriguing questions for further study.  For instance:
\begin{enumerate}
\item   It would be interesting to determine how far Theorem~\ref{main_HGT}(b) extends. For example, is the following true? For {\em any} two binary phylogenetic $X$-trees $T_1$ and $T_2$ (where $X$ can be of any size), is there an HGT network for which $T_N = T_1$ and yet where $d_N$ is representable on $T_2$  (where the mixing distribution is given by the independence model)?
\item   How do our results change if we allow some leaves to be missing (due to extinction or sampling omission)?
\item  Let $\rho(d)$ denote the minimum number of hybridizations required to represent $d$ on a hybridization or an HGT network. What conditions characterise those metrics $d$ with $\rho(d)=1$?  What about $\rho(d)=k$ for any $k\ge 1$?
\end{enumerate}

\section{Materials and Methods}
\subsection{Proof of Proposition~\ref{prophelps}}\label{app:prop2}
In the following, for a subset $q$ of $X$ of size 4 (a `quartet'), we use $T|q$ to denote the phylogenetic tree with leaf set $q$
that is induced by the $X$-tree $T$ on $q$.   Moreover, if $q=\{x,y,w,z\}$, we write $T|q = xy|wz$ if the path in $T$ connecting $x$ and $y$ is vertex-disjoint from the path in $T$ connecting $w$ and $z$.

Recall the statement of Proposition~\ref{prophelps}:

\begin{prop1} 
\mbox{ }
\begin{itemize}
\item[(a)]
Suppose that  all the trees in $\T(N)$ are isomorphic as unrooted phylogenetic $X$-trees to some tree $T$. Then $d_N$ is a tree metric that is represented by $T$.
\item[(b)]If the trees in $\T(N)$ can be partitioned into two non-empty isomorphism classes of unrooted trees, then $d_N$ is not a tree metric.
\end{itemize}
\end{prop1} 

Both parts of Proposition~\ref{prophelps} follow from the respective parts of the following Lemma~\ref{lem1}, noting that in part (b), if two trees are non-isomorphic as unrooted trees, then they must resolve at least one quartet differently.    

\begin{lem}
\label{lem1}
\mbox{ } 
Let $(T_1, w_1), (T_2, w_2), \ldots, (T_k, w_k)$ be a sequence of  phylogenetic $X$-trees with associated strictly positive edge weights.
\begin{itemize}
\item[(a)]
If $T_i=T$ for all $i$, then for any values $\beta_i\in\R^{\ge 0}$, we have:
$$\sum_{i=1}^k \beta_id_{(T, w_i)} = d_{(T, w)},$$
for the positive edge weights $\displaystyle w=\sum_i \beta_i w_i$ on $T$.
\item[(b)]
Suppose that there is a quartet $q \subseteq X$,  for which $|\{T_i|q, i=1, \ldots, k\}|=2$.  Then for any values $\beta_i\in\R^{>0 }$, we have:
$$\sum_{i=1}^2 \beta_id_{(T_i, w_i)} \neq d_{(T, w)}$$
for any phylogenetic $X$-tree $T$ having  non-negative edge weights $w$.
\end{itemize}
\end{lem}

\begin{proof}
{\em Part (a)}: By the definitions and the interchange of the order of summation we have, for any $x,y \in X$:
$$\sum_i \beta_i d_{(T, w_i)}(x,y) = \sum_i \beta_i \sum_{e \in P(T; x,y)} w_i(e) = \sum_{e \in P(T; x,y)} w(e) = d_{(T,w)}(x,y).$$

{\em Part (b)}: 
Suppose that $q=\{x,y,w,z\}$ satisfies the condition stated, with $T_j|q = xy|wz$ for all $j\in J \subseteq [k]=\{1,\ldots, k\}$, and  $T_j|q = xz|wy$ for all $j\in [k]-J$, for some  non-empty proper subset $J$ of $[k]$.
Let $d_1 = \sum_{j \in J} \beta_jd_{(T_j,  w_j)}$, 
$d_2 = \sum_{j\in [k]-J} \beta_j d_{(T_j,  w_j)}$
and $d= d_1+d_2 = \sum_{j\in [k]} \beta_j d_{(T_j,  w_j)}.$ 
By Part (a) and the four-point condition, we have:  $$d_1(x,y)+ d_1(w,z) < d_1(x,w) + d_1(y, z) = d_1(x,z)+d_1(y,w);$$
and 
$$d_2(x,z)+ d_2(w,y) < d_2(x,y) + d_2(w, z) = d_2(x,w)+d_2(y,z).$$
It follows that:
\begin{align*}
d(x,w) + d(y,z) 
  &=(d_1(x,w)+d_1(y,z))+(d_2(x,w)+d_2(y,z))\\
  &>(d_1(x,y)+d_1(w,z))+(d_2(x,y)+d_2(w,z))\\
  &=d(x,y)+d(w,z).
\end{align*}
Similarly $d(x,w) + d(y,z) >d(x,z)+d(y,w)$. Therefore,
$$d(x,w) + d(y,z) > \max\{d(x,y)+d(w,z), d(x,z)+d(y,w)\},$$ 
violating the four-point condition.
Thus $d$ has no realisation on any unrooted phylogenetic tree (binary or not) with non-negative edge weights.
\end{proof}

\subsection{Proof of Theorem~\ref{thm}}\label{app:thm2}
Recall the statement of Theorem~\ref{thm}:

\begin{thm2}
Let $X$ be a finite set of taxa, and suppose $d$ is a metric on $X$, $d:X\times X\to\R^{\ge 0}$.
\begin{enumerate}
\item[(a)]
If $d$ is a binary tree metric, then there exists a primitive 1-hybridization network $N$ and weights $w,\beta$ such that $d=d_{(N, w, \beta)}$.

\item[(b)] If $N$ is a hybridization network, and  $d=d_{(N, w, \beta)}$ is a tree metric for some $w,\beta>0$, then $N$ is either a tree, or $N$ is a primitive 1-hybridization network.
\end{enumerate}
\end{thm2} 

The proof of Theorem~\ref{thm} relies on first establishing some preliminary results.

\begin{lem}
\label{lem3}
The two binary phylogenetic $X-$trees displayed by a 1-hybridization network are isomorphic as unrooted trees if and only if the tree is a primitive 1-hybridization network.
\end{lem}
\begin{proof}
The `if' part is clear.  Conversely, suppose that a 1-hybridization network is not primitive.  Then if $(u,v)$ and $(u',v)$ denote the two reticulation arcs, there is a vertex $w$ of $N$ 
that has a leaf $z$ as a descendant (following a path of tree edges) that is not a descendant of $u, u'$ or $v$.  
Select leaves $x, x'$ and $y$ that
are descendants (following a path of tree arcs) of $u, u'$ and $v$, respectively.  Then for the two induced phylogenetic $X$-trees obtained from $N$, one tree resolves the quartet $\{x,x',y,z\}$ as $xy|x'z$, while the other tree resolves the quartet as $x'y|xz$ (as above, the vertical bar here refers to the path that connects the pair of taxa on the left being vertex-disjoint from the path connecting the pair of taxa on the right). It follows that these two induced $X$-trees are not isomorphic as unrooted trees \citep{sem}. 
\end{proof}

\begin{cor}
\label{cor:retic.near.root} 
In a 1-hybridization network, with an edge weighting $w$ and a mixing distribution $\beta$, the induced distance function $d_{(N, w, \beta)}$ is equal to $d_{(T,w)}$ for a phylogenetic $X$-tree if and only if the hybridization is between the two edges that are incident with the root. 
\end{cor}

\begin{proof}  By combining Proposition~\ref{prophelps}  and  Lemma~\ref{lem3}.
\end{proof}

We have dealt with the case in which a network has a single hybridization and shown that for it to satisfy the four-point condition (i.e. be a tree metric) the hybridization must be in a particular position, namely next to the root (the network must be primitive 1-hybridization network).  This is because both quartet trees must agree, and the only way for this to occur is if the hybridization is in this position.

To deal with the case where the network has more than one hybridization, we require a further result.

\begin{lem}
\label{lems}
If a hybridization network $N$ has four leaves and two hybridization vertices then if we consider the (at most four) trees in  $\T(N)$ and ignore their rooting, they produce exactly two unrooted quartet trees.

\end{lem}

\begin{proof}
This is an elementary check, as follows.  
There are only two 2-hybridization networks on four leaves, up to symmetry, namely those in Fig.~\ref{fig:2.retics}.

\begin{figure}[ht]
\begin{center}
\includegraphics{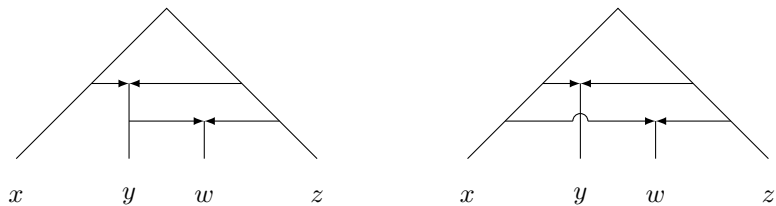}
\caption{The distinct quartets with two hybridizations, up to labelling.}\label{fig:2.retics}
\end{center}
\end{figure}

Resolving these hybridizations into the alternative unrooted quartet trees, we find that the first hybridization network only yields the unrooted quartet trees $xy|wz$ and $xz|yw$, and the second yields $xy|wz$ and $xw|yz$.  
\end{proof}

We can now prove Theorem~\ref{thm}.

\begin{proof}[Proof of Theorem~\ref{thm}]

{\em Part (a)}: Suppose $d$ is a binary tree metric. Writing $d=d_{(T,w)}$ for a binary tree $T$, select any interior vertex $v$ of $T$, and consider the three edges $e_1, e_2, e_3$
that are incident with $v$, their corresponding weights $w_1, w_2$ and $w_3$ and the rooted subtrees $T_1, T_2$ and $T_3$ that these edges are attached to, at the opposite end from $v$, as in Fig.~\ref{fig:3subtrees}$(i)$.
Let $N$ be the primitive 1-hybridization network that is obtained as follows: first, consider the rooted binary tree $T$ consisting of a root vertex attached by edges $e_1$ and $e_2$ to the roots of $T_1$
and $T_2$. Next, place reticulation arcs from (a point on each of) $e_1$ and $e_2$
to a reticulation vertex, and place a tree arc from this vertex to the root of  $T_3$.  Select any strictly positive value of $x$ with $x< \min\{w_1, w_2\}$,
and assign edge weights to $N$ as follows.  To the two edges that are incident with the root vertex assign weight $x$; to the tree arcs that are incident with $T_1$ and $T_2$ assign weights
$w_1-x$ and $w_2-x$, respectively, and to the tree arcs incident with $T_3$, assign the weight $w_3$. To the reticulation arcs assign a uniform hybridization distribution ($\alpha_1=\alpha_2=\frac{1}{2}$) (see Fig. ~\ref{fig:3subtrees}$(ii)$).  Then it can be checked that $N$, together with this arc weighting and hybridization distribution, gives a distance function that coincides exactly with $d$.

\begin{figure}[ht]
\centering
\includegraphics{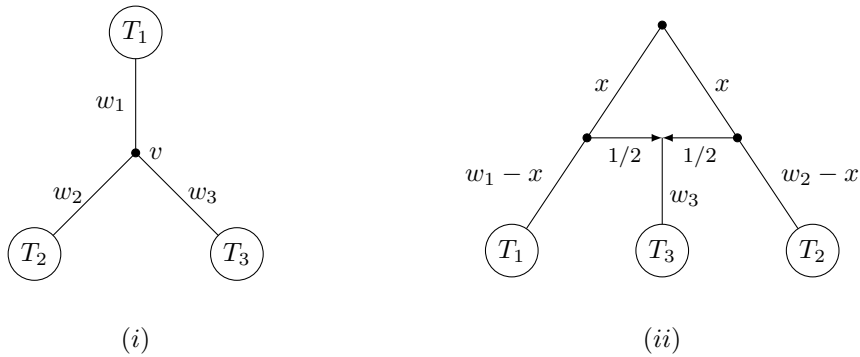}
\caption{$T_i$ represents rooted subtrees of $T$.  As in the proof of Theorem~\ref{thm} Part~(a), we have an unrooted tree and a primitive 1-reticulated network with the same average distance function.}
\label{fig:3subtrees}
\end{figure}

{\em Part (b)}:   
Suppose a network $N$ has more than one hybridization.  
We first consider the case where there is at least one reticulation arc $(u,v)$ whose source vertex ($u$) lies below a non-root vertex $w$ that has two outgoing tree arcs.  
Let  $(u',v)$ be the other reticulation arc of $N$ that ends at reticulation vertex $v$.
We will construct a quartet of leaves that give rise to a non-primitive 1-hybridization network.

Let $a, b$ and $c$ be three leaves obtained by following tree arcs from $u, v$ and $u'$, respectively (every internal vertex has at least one outgoing tree arc by the definition of a reticulation network, part (v)).  
Choose a fourth leaf $d$ that is reached by an arc from the root that does not pass through $u$ or $u'$.  This can always be done; consider two cases.  First, if $w$ lies above $u'$ as well as $u$, then paths from the root that do not go through $w$ will also not go through $u$ or $u'$. Alternatively, there is a path from the root through $w$ that goes down the ``other'' tree arc from $w$ (the one not leading to $u$) that will not pass through $u'$.  The restrictions of $N$ to $\{a,b,c,d\}$ in these two cases are shown in Figure~\ref{fig:placing.d}.
In either case, the restriction is a 1-hybridisation network that is not primitive, and so, by Corollary~\ref{cor:retic.near.root}, does not induce a tree metric on $\{a,b,c,d\}$. 
Thus $d=d_{(N, w, \beta)}$ cannot be a tree metric on $X$.
\begin{figure}[ht]
\includegraphics{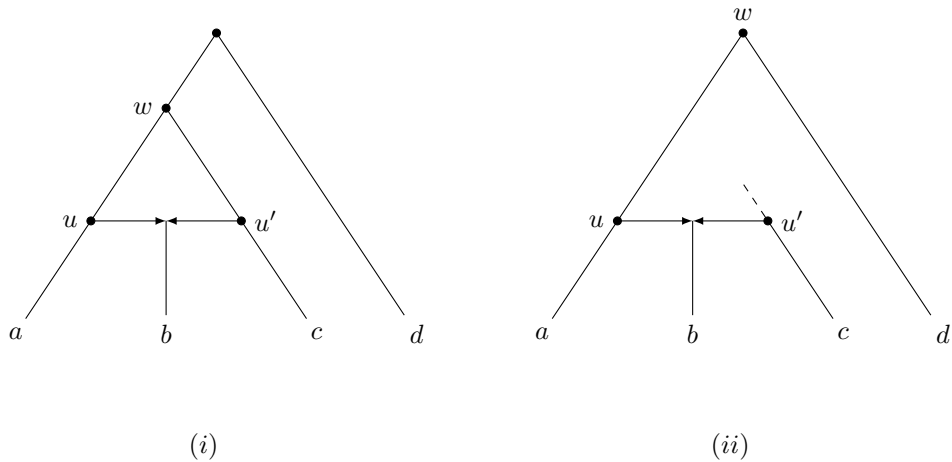} 
\caption{The two cases arising in Part (b) of the proof of Theorem~\ref{thm}, showing the restriction of $N$ to $\{a,b,c,d\}$.  $(i)$ shows the case that $w$ is above both $u$ and $u'$, and $(ii)$ shows the case $w$ is above $u$ but not $u'$.  
}\label{fig:placing.d}
\end{figure}

Thus, we may suppose that if $N$ has more than one hybridisation, then none of the source vertices of any reticulation arc lie below any non-root vertex that has two outgoing tree arcs.
That is, the source vertices of all reticulation arcs lie below either the root, or another source vertex of a reticulation arc (if a non-root vertex does not have two outgoing tree arcs, then it must have a reticulation arc).  Another way to view this is that as one proceeds along any path from the root to a leaf, once one encounters a tree vertex one never encounters another reticulation vertex.
This forces the reticulation vertices to be near the root, and for there to be a quartet in which at least two hybridizations appear (see Figure~\ref{fig:2.retics}). Therefore, since $r\le n-2$ (Inequality~\eqref{counter2}), exactly two hybridizations occur.  

Such a quartet can be chosen simply by a suitable choice of leaves.   By Lemma~\ref{lems} and Proposition~\ref{prophelps} this implies that
$d_N$ restricted to this quartet is not a tree metric on that quartet, which violates the assumption that $d_N$ is a tree metric on all of $X$.   Thus, $N$ must be a 1-hybridization network.  Lemma~\ref{lem3} and Proposition~\ref{prophelps}(b) now imply that $N$ must also be primitive.  This completes the proof.
\end{proof}

\section{Acknowledgments}

We thank the two anonymous reviewers for some helpful comments on an earlier version of this manuscript. The first author thanks the Australian Research Council (via FT100100898) and the second author thanks the NZ Marsden Fund and Allan Wilson Centre for helping to fund this work.


\begin{thebibliography}{18}
\providecommand{\natexlab}[1]{#1}
\providecommand{\url}[1]{\texttt{#1}}
\providecommand{\urlprefix}{URL }
\expandafter\ifx\csname urlstyle\endcsname\relax
  \providecommand{\doi}[1]{doi:\discretionary{}{}{}#1}\else
  \providecommand{\doi}[1]{doi:\discretionary{}{}{}\begingroup
  \urlstyle{rm}\url{#1}\endgroup}\fi
\providecommand{\bibinfo}[2]{#2}

\bibitem[{Felsenstein(2004)}]{fel}
\bibinfo{author}{J.~Felsenstein}, \bibinfo{title}{Inferring Phylogenies},
  \bibinfo{publisher}{Sinauer Press}, \bibinfo{year}{2004}.

\bibitem[{McBreen and Lockhart(2006)}]{mcb}
\bibinfo{author}{K.~McBreen}, \bibinfo{author}{P.~L. Lockhart},
  \bibinfo{title}{Reconstructing reticulate evolutionary histories of plants},
  \bibinfo{journal}{Trends Plant Sci.}
  \bibinfo{volume}{11}~(\bibinfo{number}{8}) (\bibinfo{year}{2006})
  \bibinfo{pages}{398--404}.

\bibitem[{Huson et~al.(2010)Huson, Rupp, and Scornavacca}]{hus}
\bibinfo{author}{D.~H. Huson}, \bibinfo{author}{R.~Rupp},
  \bibinfo{author}{C.~Scornavacca}, \bibinfo{title}{Phylogenetic Networks},
  \bibinfo{publisher}{Cambridge University Press}, \bibinfo{year}{2010}.

\bibitem[{Huson and Bryant(2006)}]{hus2}
\bibinfo{author}{D.~H. Huson}, \bibinfo{author}{D.~Bryant},
  \bibinfo{title}{Application of phylogenetic networks in evolutionary
  studies}, \bibinfo{journal}{Mol. Biol. Evol.} \bibinfo{volume}{23}
  (\bibinfo{year}{2006}) \bibinfo{pages}{254--267}.

\bibitem[{Nakhleh et~al.(2005)Nakhleh, Warnow, and Linder}]{nak}
\bibinfo{author}{L.~Nakhleh}, \bibinfo{author}{T.~Warnow},
  \bibinfo{author}{C.~R. Linder}, \bibinfo{title}{Reconstructing reticulate
  evolution in species: theory and practice}, \bibinfo{journal}{J. Comput.
  Biol.} \bibinfo{volume}{12} (\bibinfo{year}{2005}) \bibinfo{pages}{796--811}.

\bibitem[{Holder et~al.(2001)Holder, Anderson, and Holloway}]{hold}
\bibinfo{author}{M.~T. Holder}, \bibinfo{author}{J.~A. Anderson},
  \bibinfo{author}{A.~K. Holloway}, \bibinfo{title}{Difficulties in detecting
  hybridization}, \bibinfo{journal}{Syst. Biol.}
  \bibinfo{volume}{50}~(\bibinfo{number}{6}) (\bibinfo{year}{2001})
  \bibinfo{pages}{978--982}.

\bibitem[{Holland et~al.(2008)Holland, Bentham, Lockhart, Moulton, and
  Huber}]{holl}
\bibinfo{author}{B.~Holland}, \bibinfo{author}{S.~Bentham},
  \bibinfo{author}{P.~J. Lockhart}, \bibinfo{author}{V.~Moulton},
  \bibinfo{author}{K.~T. Huber}, \bibinfo{title}{The power of supernetworks to
  distinguish hybridization from lineage-sorting via collections of gene
  trees}, \bibinfo{journal}{BMC Evol. Biol.} \bibinfo{volume}{8}
  (\bibinfo{year}{2008}) \bibinfo{pages}{202},
  \doi{\bibinfo{doi}{10.1186/1471-2148-8-202}}.

\bibitem[{Linz et~al.(2010)Linz, Semple, and Stadler}]{linz}
\bibinfo{author}{S.~Linz}, \bibinfo{author}{C.~Semple},
  \bibinfo{author}{T.~Stadler}, \bibinfo{title}{Analyzing and reconstructing
  reticulation networks under timing constraints}, \bibinfo{journal}{J. Math.
  Biol.} \bibinfo{volume}{61} (\bibinfo{year}{2010}) \bibinfo{pages}{715--737}.

\bibitem[{Bullini(1994)}]{bul}
\bibinfo{author}{L.~Bullini}, \bibinfo{title}{Origin and evolution of animal
  hybrid species}, \bibinfo{journal}{Trends Ecol. Evol.}
  \bibinfo{volume}{9}~(\bibinfo{number}{11}) (\bibinfo{year}{1994})
  \bibinfo{pages}{422--426}.

\bibitem[{Dagan et~al.(2008)Dagan, Artzy-Randrup, and Martin}]{dag}
\bibinfo{author}{T.~Dagan}, \bibinfo{author}{Y.~Artzy-Randrup},
  \bibinfo{author}{W.~Martin}, \bibinfo{title}{Modular networks and cumulative
  impact of lateral transfer in prokaryote genome evolution},
  \bibinfo{journal}{Proc. Natl. Acad. Sci. USA} \bibinfo{volume}{105}
  (\bibinfo{year}{2008}) \bibinfo{pages}{10039--10044}.

\bibitem[{Semple and Steel(2003)}]{sem}
\bibinfo{author}{C.~Semple}, \bibinfo{author}{M.~Steel},
  \bibinfo{title}{Phylogenetics}, \bibinfo{publisher}{Oxford University Press},
  \bibinfo{year}{2003}.

\bibitem[{Cardona et~al.(2009)Cardona, Rossell{\'o}, and Valiente}]{car}
\bibinfo{author}{G.~Cardona}, \bibinfo{author}{G.~Rossell{\'o}},
  \bibinfo{author}{G.~Valiente}, \bibinfo{title}{Comparison of tree-child
  phylogenetic networks}, \bibinfo{journal}{IEEE/ACM Trans. Comput. Biol.
  Bioinf.} \bibinfo{volume}{6}~(\bibinfo{number}{4}) (\bibinfo{year}{2009})
  \bibinfo{pages}{552--569}.

\bibitem[{van Iersel et~al.(2010)van Iersel, Semple, and Steel}]{van}
\bibinfo{author}{L.~van Iersel}, \bibinfo{author}{C.~Semple},
  \bibinfo{author}{M.~Steel}, \bibinfo{title}{Locating a tree in a phylogenetic
  network}, \bibinfo{journal}{Inf. Process. Lett.} \bibinfo{volume}{110}
  (\bibinfo{year}{2010}) \bibinfo{pages}{1037--1043}.

\bibitem[{McDiarmid et~al.(2014)McDiarmid, Semple, and Welsh}]{mcd}
\bibinfo{author}{C.~McDiarmid}, \bibinfo{author}{C.~Semple},
  \bibinfo{author}{D.~Welsh}, \bibinfo{title}{Counting phylogenetic networks},
  \bibinfo{journal}{Ann. Comb.} \bibinfo{volume}{(in press)}.

\bibitem[{Dress et~al.(2012)Dress, Huber, Koolen, Moulton, and Spillner}]{dre}
\bibinfo{author}{A.~Dress}, \bibinfo{author}{K.~T. Huber},
  \bibinfo{author}{J.~Koolen}, \bibinfo{author}{V.~Moulton},
  \bibinfo{author}{A.~Spillner}, \bibinfo{title}{Basic Phylogenetic
  Combinatorics}, \bibinfo{publisher}{Cambridge University Press},
  \bibinfo{year}{2012}.

\bibitem[{Willson(2012)}]{willson2012}
\bibinfo{author}{S.~J. Willson}, \bibinfo{title}{Tree-average distances on
  certain phylogenetic networks have their weights uniquely determined},
  \bibinfo{journal}{Algorithms Mole. Biol.} \bibinfo{volume}{7}
  (\bibinfo{year}{2012}) \bibinfo{pages}{13},
  \doi{\bibinfo{doi}{10.1186/1748-7188-7-13}}.

\bibitem[{Willson(2013)}]{willson2013}
\bibinfo{author}{S.~Willson}, \bibinfo{title}{Reconstruction of certain
  phylogenetic networks from their tree-average distances},
  \bibinfo{journal}{Bull. Math. Biol.}
  \bibinfo{volume}{75}~(\bibinfo{number}{10}) (\bibinfo{year}{2013})
  \bibinfo{pages}{1840--1878}.

\bibitem[{Durrett(2008)}]{dur}
\bibinfo{author}{R.~Durrett}, \bibinfo{title}{Probability models of {D}{N}{A}
  sequence evolution}, \bibinfo{publisher}{Springer, 2nd. Ed.},
  \bibinfo{year}{2008}.

\end{thebibliography}

\end{document}